\documentclass{svjour3}
\def\makeheadbox{}
\smartqed  
\usepackage[all]{xy}
\SelectTips{cm}{}

\usepackage{url,amssymb,mathptmx,amsmath}
\DeclareFontFamily{OML}{cyi}{} \DeclareFontShape{OML}{cyi}{m}{n}{
  <5> <6> <7> <8> <9> gen * wncyi
  <10> <10.95> <12> <14.4> <17.28> <20.74> <24.88> wncyi10
 }{}
\DeclareSymbolFont{rusletters}{OML}{cyi}{m}{n}
\DeclareSymbolFontAlphabet{\rusmath}{rusletters}
\DeclareMathSymbol\re{\rusmath}{rusletters}{"03}
\renewcommand*{\d}{\mathinner{\!}d}
\def\cprime{\/{\mathsurround=0pt$'$}}
\DeclareMathOperator{\sym}{sym}
\DeclareMathOperator{\cosym}{cosym}
\DeclareMathOperator{\Hom}{Hom}
\newcommand*{\eval}[1]{\left.#1\right|}
\newcommand*{\abs}[1]{\left|#1\right|}
\newcommand{\ldb}{[\![}
\newcommand{\rdb}{]\!]}
\newcommand*{\pd}[2]{\frac{\partial#1}{\partial#2}}
\let\kappa=\varkappa
\let\phi=\varphi
\allowdisplaybreaks[4]
\journalname{Acta Applicandae Mathematicae}
\begin{document}


\title{On integrability of the Camassa--Holm equation and its
  invariants\thanks{This work was supported in part by the NWO--RFBR grant
    047.017.015 and RFBR--Consortium E.I.N.S.T.E.I.N.\ grant 06-01-92060.}  }
\subtitle{A geometrical approach}

\titlerunning{Integrability of the Camassa--Holm equation}

\author{V.~Golovko \and P.~Kersten \and \\
I.~Krasil{\cprime}shchik \and A.~Verbovetsky }

\authorrunning{V.~Golovko et al.}

\institute{Valentina Golovko \at Lomonosov MSU, Faculty of Physics, Department
  of Mathematics,
  Vorob{\cprime}evy Hills, Moscow 119902, Russia\\
  \email{golovko@mccme.ru} 
  \and Paul Kersten \at University of Twente, Postbus 217, 7500 AE Enschede,
  the Netherlands\\
  \email{kerstenphm@ewi.utwente.nl} 
  \and Iosif Krasil{\cprime}shchik \at Independent University of Moscow,
  B. Vlasevsky~11,
  119002 Moscow,  Russia\\
  \email{josephk@diffiety.ac.ru} 
  \and Alexander Verbovetsky \at Independent University of Moscow,
  B. Vlasevsky~11,
  119002 Moscow,  Russia\\
  \email{verbovet@mccme.ru} }

\date{\ }

\maketitle

\begin{abstract}
  Using geometrical approach exposed in Refs.~\cite{K-K-V-lstar,K-K-V-dB}, we
  explore the Camassa--Holm equation (both in its initial scalar form, and in
  the form of $2\times2$-system). We describe Hamiltonian and symplectic
  structures, recursion operators and infinite series of symmetries and
  conservation laws (local and nonlocal).  
\end{abstract}

\section{Introduction}
\label{sec:introduction}
The Camassa--Holm equation was introduced in~\cite{C-H} in the form
\begin{equation}
  \label{eq:1}
  u_t+\mu u_x-u_{xxt}+3uu_x=2u_xu_{xx}+uu_{xxx},\qquad\mu\in\mathbb{R},
\end{equation}
and was intensively explored afterwards (see, for example,
Refs.~\cite{C-L-O-P,C-I,F-S,L}). Its superizations were also constructed,
see~\cite{A-G-Z,P}. Since~\eqref{eq:1} is not an evolution equation, its
integrability properties (existence and even definition of Hamiltonian
structures, conservation laws, etc.) are not standard to establish.

One of the ways widely used to overcome this difficulty is to introduce a new
unknown~$m=u-u_{xx}$ and transform Eq.~\eqref{eq:1} to the system
\begin{equation}
  \label{eq:2}
  \begin{cases}
    m_t=-um_x-(2m+\mu)u_x,\\
    m=u-u_{xx},
  \end{cases}
\end{equation}
which has \emph{almost} evolutionary form. We stress this ``almost'', because
the second equation in~\eqref{eq:2} (that can be considered as a constrain to
the first one) disrupts the picture and, at best, necessitates to invert the
operator~$1-D_x^2$. At worst, dealing with Eq.~\eqref{eq:2} as with an
evolution equation may lead to fallacious results.

In our approach based on the geometrical framework exposed in
Ref.~\cite{B-Ch-D-AMS}, we treat the equation at hand as a submanifold in the
manifold of infinite jets and consider two natural extensions of this
equation, cf.~with Ref~\cite{K}. The first one is called the
\emph{$\ell$}-covering and serves the role of the tangent bundle. The second
extension, $\ell^*$-covering, is the counterpart to the cotangent bundle. The
key property of these extensions is that the spaces of their nonlocal (in the
sense of~\cite{K-V}) symmetries and cosymmetries contain all essential
integrability invariants of the initial equation. The efficiency of the method
was tested for a number of problems (see
Refs.~\cite{K-K-V-lstar,K-K-V-dB,K-K-V-sup}) and we apply it to the
Camassa--Holm equation here.

In Section~\ref{sec:underlying-theory} we briefly expose the necessary
definition and facts. Section~\ref{sec:matrix-version} contains computations
for the Camassa--Holm equation in its matrix version (computations and results
are more compact in this representation), while in
Section~\ref{sec:scalar-version} we reformulate them for the original
form~\eqref{eq:1} and compare later the results obtained for the two
alternative presentations. Finally, Section~\ref{sec:conclusion} contains
discussion of the results obtained. Throughout our exposition we use a very
stimulating conceptual parallel between categories of smooth manifolds and
differential equations proposed initially by A.M.~Vinogradov and in its modern
form presented in Table~\ref{tab:1}.

\begin{table}[h]
\caption{Conceptual parallel between two categories}
\label{tab:1}
\begin{tabular}{l|l}
\hline\noalign{\smallskip}
\textbf{Manifolds} & \textbf{Equation}s\\
\noalign{\smallskip}\hline\noalign{\smallskip}
Smooth manifold & Infinitely prolonged equation\\
Point & Formal solution\\
Smooth function & Conservation law\\
Vector field & Higher symmetry\\
Differential $1$-form & Cosymmetry\\
de~Rham complex & $(n-1)$st line of Vinogradov's $\mathcal{C}$-spectral sequence\\
Tangent bundle & $\ell$-covering \\
Cotangent bundle & $\ell^*$-covering \\
\noalign{\smallskip}\hline
\end{tabular}
\end{table}
\noindent

This table is not just a toy dictionary but a quite helpful tool to formulate
important definitions and results. For example, a bivector on a smooth
manifold~$M$ may be understood as a derivation of the ring~$C^\infty(M)$ with
values in~$C^\infty(T^*M)$. Translating this statement to the language of
differential equations we come to the definition of variational bivectors and
their description as shadows of symmetries in the $\ell^*$-covering (see
Theorem~\ref{thm:ellstar-shad} below). Another example: any vector field
(differential $1$-form) on~$M$ may be treated as a function on~$T^*M$
(on~$TM$). Hence, to any symmetry (cosymmetry) there corresponds a
conservation law on the space of the $\ell^*$-covering ($\ell$-covering). This
leads to the notions of nonlocal vectors and forms that, in turn, provide a
basis to construct weakly nonlocal structures (see
Subsections~\ref{sec:nonlocal-forms-1} and~\ref{sec:nonlocal-vectors-1}). Of
course, these parallels are not completely straightforward (in technical
aspects, especially), but extremely enlightening and fruitful.

The idea of this paper arose in the discussions one of the authors had with
Volodya Roubtsov in 2007. We agreed to write two parallel texts on
integrability of the Camassa--Holm equation that reflect our viewpoints. The
reader can now compare our results with the ones presented in~\cite{C-O-R}.

\section{Underlying theory}
\label{sec:underlying-theory}

We present here a concise exposition of the theoretical background used in
the subsequent sections, see Refs.~\cite{B-Ch-D-AMS,K-K-V-dB,K-V}.

\subsection{Equations, symmetries, etc.}
\label{sec:equat-symm-etc}

Let~$\pi\colon E\to M$ be a fiber bundle and~$\pi_\infty\colon
J^\infty(\pi)\to M$ be the bundle of its infinite jets. To simplify our
exposition we shall assume that~$\pi$ is a vector bundle. In all applications
below~$\pi$ is the trivial
bundle~$\mathbb{R}^m\times\mathbb{R}^n\to\mathbb{R}^n$. We consider
\emph{infinite prolongations} of differential equations as
submanifolds~$\mathcal{E}\subset J^\infty(\pi)$ and retain the
notation~$\pi_\infty$ for the
restriction~$\eval{\pi_\infty}_{\mathcal{E}}$. Any such a manifold is endowed
with the \emph{Cartan distribution} which spans at every point tangent spaces
to the graphs of jets. A \emph{symmetry} of~$\mathcal{E}$ is a vector field
that preserves this distribution. The set of symmetries is a Lie algebra
over~$\mathbb{R}$ denoted by~$\sym\mathcal{E}$.

For any equation~$\mathcal{E}$ its \emph{linearization operator}
$\ell_{\mathcal{E}}\colon\kappa\to P$ is defined, where~$\kappa$ is the
module\footnote{All the modules below are modules over the ring~$\mathcal{F}$
  of smooth functions on~$\mathcal{E}$.}
of sections of the pullback~$\pi_\infty(\pi)$ and~$P$ is the module of
sections of some vector bundle over~$\mathcal{E}$. Then~$\sym\mathcal{E}$ can
be identified with solutions of the equation
\begin{equation}
  \label{eq:3}
  \ell_{\mathcal{E}}(\phi)=0,\qquad\phi\in\kappa.
\end{equation}
For two symmetries~$\phi_1$, $\phi_2\in\sym\mathcal{E}$ their commutator is
denoted by~$\{\phi_1,\phi_2\}$.

Denote by~$\Lambda_h^i$ the module of horizontal $i$-forms on~$\mathcal{E}$
and introduce the notation
\begin{equation*}
  \hat{Q}=\Hom_{\mathcal{F}}(Q,\Lambda_h^n),\qquad n=\dim M,
\end{equation*}
for any module~$Q$. The \emph{adjoint} to~$\ell_{\mathcal{E}}$ operator
\begin{equation}
  \label{eq:4}
  \ell_{\mathcal{E}}^*\colon\hat{P}\to\hat{\kappa}
\end{equation}
arises and solutions of the equation
\begin{equation}
  \label{eq:5}
  \ell_{\mathcal{E}}^*(\psi)=0,\qquad\psi\in\hat{P},
\end{equation}
are called \emph{cosymmetries} of~$\mathcal{E}$; the space of cosymmetries is
denoted by~$\cosym\mathcal{E}$.

Let~$d_h\colon\Lambda_h^i\to\Lambda_h^{i+1}$ be the \emph{horizontal de~Rham
  differential}. A \emph{conservation law} of the equation~$\mathcal{E}$ is a
closed form~$\omega\in\Lambda_h^{n-1}$. To any conservation law there
corresponds its \emph{generating function}~$\delta\omega\in\cosym\mathcal{E}$,
where~$\delta\colon E_1^{0,n-1}\to E_1^{1,n-1}$ is the differential in the
$E_1$ term of Vinogradov's $\mathcal{C}$-spectral sequence, see~\cite{V}. In
the evolutionary case~$\delta$ coincides with the Euler--Lagrange operator. A
conservation law is \emph{trivial} if its generating function vanishes. In
particular, $d_h$-exact conservation laws are trivial.

A vector field on~$\mathcal{E}$ is called a \emph{$\mathcal{C}$-field} if it
lies in the Cartan distribution. A differential operator~$\Delta\colon P\to
Q$, $P$ and~$Q$ being $\mathcal{F}$-modules, is called a
\emph{$\mathcal{C}$-differential operator} if it is locally expressed in terms
of $\mathcal{C}$-fields. For example,~$\ell_{\mathcal{E}}$ is a
$\mathcal{C}$-differential operator.

A $\mathcal{C}$-differential operator~$\tens{H}\colon\hat{P}\to\kappa$ is said
to be a \emph{variational bivector} on~$\mathcal{E}$ if
\begin{equation}
  \label{eq:6}
  \ell_{\mathcal{E}}\circ\tens{H}=\tens{H}^*\circ\ell_{\mathcal{E}}^*.
\end{equation}
This condition means that~$\tens{H}$ takes cosymmetries of~$\mathcal{E}$ to
symmetries. If~$\mathcal{E}$ is an evolution Eq.~\eqref{eq:6} implies also
that~$\tens{H}^*=-\tens{H}$. A bivector~$\tens{H}$ is a \emph{Hamiltonian
  structure}\footnote{It is more appropriate to call these objects
  \emph{Poisson structures}, but we follow the tradition accepted in the
  theory of integrable systems.} on~$\mathcal{E}$
if~$\ldb\tens{H},\tens{H}\rdb^{\mathrm{s}}=0$,
where~$\ldb\cdot,\cdot\rdb^{\mathrm{s}}$ is the \emph{variational Schouten
  bracket} (see also~\cite{I-V-V}). Two Hamiltonian structures are
\emph{compatible} if~$\ldb\tens{H}_1,\tens{H}_2\rdb^{\mathrm{s}}=0$.

A $\mathcal{C}$-differential operator~$\tens{S}\colon\kappa\to\hat{P}$ is
called a \emph{variational $2$-form on~$\mathcal{E}$} if
\begin{equation}
  \label{eq:7}
  \ell_{\mathcal{E}}^*\circ\tens{S}=\tens{S}^*\circ\ell_{\mathcal{E}}.
\end{equation}
Such operators take symmetries to cosymmetries and in evolutionary case are
skew-adjoint. They are elements of the term~$E_1^{2,n-1}$ of Vinogradov's
$\mathcal{C}$-spectral sequence. A variational form is a \emph{symplectic
  structure} on the equation~$\mathcal{E}$ if it is \emph{variationally
  closed}, i.e.,~$\delta\tens{S}=0$, where~$\delta\colon E_1^{2,n-1}\to
E_1^{3,n-1}$ is the corresponding differential.

We shall also consider \emph{recursion $\mathcal{C}$-differential
  operators}~$\tens{R}\colon\kappa\to\kappa$
and~$\hat{\tens{R}}\colon\hat{P}\to\hat{P}$ satisfying the conditions
\begin{equation}
  \label{eq:8}
  \ell_{\mathcal{E}}\circ\tens{R}=\tens{R}'\circ\ell_{\mathcal{E}},\qquad
  \ell_{\mathcal{E}}^*\circ\hat{\tens{R}}=\hat{\tens{R}}'\circ\ell_{\mathcal{E}}^*
\end{equation}
for some $\mathcal{C}$-differential operators~$\tens{R}'\colon P\to P$
and~$\hat{\tens{R}}'\colon\hat{\kappa}\to\hat{\kappa}$. An operator~$\tens{R}$
satisfies the \emph{Nijenhuis condition}
if~$\ldb\tens{R},\tens{R}\rdb^{\mathrm{n}}=0$,
where~$\ldb\cdot,\cdot\rdb^{\mathrm{n}}$ is the \emph{variational Nijenhuis
  bracket}. A recursion operator~$\tens{R}$ is \emph{compatible} with a
Hamiltonian structure~$\tens{H}$ is the composition~$\tens{R}\circ\tens{H}$ is
a Hamiltonian structure as well.

\subsection{Nonlocal theory}
\label{sec:nonlocal-theory}

Let~$\mathcal{E}$ and~$\tilde{\mathcal{E}}$ be equations
and~$\xi\colon\tilde{\mathcal{E}}\to\mathcal{E}$ be a fiber bundle. Denote
by~$\mathcal{C}$ and~$\tilde{\mathcal{C}}$ the Cartan distributions
on~$\mathcal{E}$ and~$\tilde{\mathcal{E}}$, resp. We say that~$\xi$ is a
\emph{covering} if for any~$\tilde{\theta}\in\tilde{\mathcal{E}}$ the
differential~$d_{\tilde{\theta}}\xi$ isomorphically
maps~$\tilde{\mathcal{C}}_{\tilde{\theta}}$
onto~$\mathcal{C}_{\xi(\tilde{\theta})}$. A particular case of coverings (the
so-called \emph{Abelian} coverings) is naturally associated with closed
horizontal $1$-forms\footnote{\label{fn:1}When~$\dim M=2$, Abelian coverings
  are associated with conservation laws of the equation~$\mathcal{E}$.}.

By definition, any $\mathcal{C}$-field~$X$ on~$\mathcal{E}$ can be uniquely
lifted to a $\mathcal{C}$-field~$\tilde{X}$ on~$\tilde{\mathcal{E}}$ such
that~$d\xi(\tilde{X})=X$. Consequently, any $\mathcal{C}$-differential
operator~$\Delta\colon P\to Q$ is extended to a $\mathcal{C}$-differential
operator
\begin{equation*}
  \tilde{\Delta}\colon\tilde{\mathcal{F}}\otimes_{\mathcal{F}}P\to
  \tilde{\mathcal{F}}\otimes_{\mathcal{F}}Q,
\end{equation*}
$\tilde{\mathcal{F}}$ being the algebra of smooth functions
on~$\tilde{\mathcal{E}}$.

A \emph{nonlocal} $\xi$-(co)symmetry of~$\mathcal{E}$ is a (co)symmetry of the
covering equation~$\tilde{\mathcal{E}}$. They are solutions of the
equations~$\ell_{\tilde{\mathcal{E}}}\phi=0$
and~$\ell_{\tilde{\mathcal{E}}}^*\psi=0$, resp. Along with these two equations
one can consider the equations
\begin{equation}
  \label{eq:9}
  (1)\ \ \tilde{\ell}_{\mathcal{E}}\phi=0,\qquad
  (2)\ \ \tilde{\ell}_{\mathcal{E}}^*\psi=0.
\end{equation}
Their solutions are called \emph{$\xi$-shadows} of symmetries and
cosymmetries, resp. A shadow of symmetry is a
derivation~$\mathcal{F}\to\tilde{\mathcal{F}}$ that preserves the Cartan
distributions. For any two shadows of symmetries~$\phi_1$ and~$\phi_2$ their
commutator~$\{\phi_1,\phi_2\}$ can be defined. This commutator is a shadow in
a new covering that is canonically determined by~$\phi_1$ and~$\phi_2$.

\subsection{The $\ell$- and $\ell^*$-coverings}
\label{sec:ell-ell-coverings}

Let~$\mathcal{E}\subset J^\infty(\pi)$ be an equation. Its
\emph{$\ell$-covering}~$\tau\colon\mathcal{L}(\mathcal{E})\to\mathcal{E}$ is
obtained by adding to~$\mathcal{E}$ the equation~$\ell_{\mathcal{E}}(q)=0$,
where~$q$ is a new variable. Dually, the
\emph{$\ell^*$-covering}~$\tau^*\colon\mathcal{L}^*(\mathcal{E})\to\mathcal{E}$
is constructed by adding the equation~$\ell_{\mathcal{E}}^*(p)=0$ with a new
variable~$p$. They are the exact counterparts of the tangent and cotangent
bundles in the category of differential equations. By the reasons that will
become clear later, we regard both~$q$ and~$p$ as \emph{odd} variables. The
main point of our method is the fundamental relation between integrability
invariants of~$\mathcal{E}$ and shadows in~$\tau$ and~$\tau^*$. To formulate
this relation, let us give an auxiliary definition: for an arbitrary operator
equation~$A\circ\Delta=\nabla\circ B$ we say that~$\Delta$ is a \emph{trivial
  solutions} if~$\Delta$ is of the form~$\Delta=\Delta'\circ B$. Classes of
solutions modulo trivial ones will be called \emph{nontrivial solutions}. Then
the following results hold.

\begin{theorem}[shadows in the $\ell$-covering]
  \label{thm:ell-shad}
  There is a one-to-one correspondence between nontrivial solutions of the
  equation
  \begin{equation*}
    \ell_{\mathcal{E}}\circ\tens{R}=\tens{R}'\circ\ell_{\mathcal{E}},\qquad
    \tens{R}\colon\kappa\to\kappa,
  \end{equation*}
  and $\tau$-shadows of symmetries linear w.r.t.\ the variables~$q$. In a
  similar way\textup{,} there is a one-to-one correspondence between
  nontrivial solutions of the equation
  \begin{equation*}
    \ell_{\mathcal{E}}^*\circ\tens{S}=\tens{S}'\circ\ell_{\mathcal{E}},\qquad
    \tens{S}\colon\kappa\to\hat{P},
  \end{equation*}
  and $\tau$-shadows of cosymmetries linear w.r.t.\ the variables~$q$.
\end{theorem}

\begin{theorem}[shadows in the $\ell^*$-covering]
  \label{thm:ellstar-shad}
  There is a one-to-one correspondence between nontrivial solutions of the
  equation
  \begin{equation*}
    \ell_{\mathcal{E}}\circ\tens{H}=\tens{H}'\circ\ell_{\mathcal{E}}^*,\qquad
    \tens{H}\colon\hat{P}\to\kappa,
  \end{equation*}
  and $\tau^*$-shadows of symmetries linear w.r.t.\ the variables~$p$. In a
  similar way\textup{,} there is a one-to-one correspondence between
  nontrivial solutions of the equation
  \begin{equation*}
    \ell_{\mathcal{E}}^*\circ\tens{\hat{R}}=\hat{\tens{R}}'\circ\ell_{\mathcal{E}}^*,
    \qquad
    \hat{\tens{R}}\colon\hat{P}\to\hat{P},
  \end{equation*}
  and $\tau^*$-shadows of cosymmetries linear w.r.t.\ the variables~$p$.
\end{theorem}

\begin{theorem}
  \label{thm:brackets}
  Let~$\tens{R}_1$ and~$\tens{R}_2$ be recursion operators for symmetries
  on~$\mathcal{E}$ and~$\Phi_{\tens{R}}$ denote the $\tau$-shadow
  corresponding
  to~$\tens{R}$. Then~$\ldb\tens{R}_1,\tens{R}_2\rdb^{\mathrm{n}}=0$
  iff~$\{\Phi_{\tens{R}_1},\Phi_{\tens{R}_1}\}=0$. Similarly\textup{,}
  if~$\tens{H}_1$ and~$\tens{H}_2$ are bivectors
  then~$\ldb\tens{H}_1,\tens{H}_2\rdb^{\mathrm{s}}=0$
  iff~$\{\Phi_{\tens{H}_1}^*,\Phi_{\tens{H}_1}^*\}=0$\textup{,}
  where~$\Phi_{\tens{H}}^*$ denotes the $\tau^*$-shadow corresponding
  to~$\tens{H}$.
\end{theorem}

In both cases the curly brackets denote the \emph{super} bracket of shadows
that arises due to oddness of the variables~$q$ and~$p$. Additional discussion
of Theorem~\ref{thm:brackets} the reader will find in
Remark~\ref{rem:nonlocal}.

\begin{theorem}
  \label{thm:correpondence}
  To any cosymmetry of~$\mathcal{E}$ there canonically corresponds a
  conservation law of~$\mathcal{L}(\mathcal{E})$. Dually\textup{,} to any
  symmetry of~$\mathcal{E}$ there canonically corresponds a conservation law
  of~$\mathcal{L}^*(\mathcal{E})$.
\end{theorem}

\subsection{Computational scheme}
\label{sec:computational-scheme}

Let locally the equation~$\mathcal{E}$ be given by the system
\begin{equation}
  \label{eq:10}
  \begin{cases}
    F^1(x^1,\dots,x^n,\dots,\pd{^{\abs{\sigma}}u^j}{x^\sigma},\dots)=0,\\
    \dots\\
    F^r(x^1,\dots,x^n,\dots,\pd{^{\abs{\sigma}}u^j}{x^\sigma},\dots)=0,
  \end{cases}
\end{equation}
where~$j=1,\dots,m$ and~$\abs{\sigma}\le k$.

\paragraph{Step 1}
\label{sec:step-1}
consists of writing out defining equations for symmetries and cosymmetries
of~$\mathcal{E}$. Let~$\{u_\sigma^j\}_{\sigma\in S}^{j\in J}$ be
\emph{internal coordinates} on~$\mathcal{E}$, $S$ and~$J$ being some sets of
(multi)indices and~$u_\sigma^j$ corresponding
to~$\partial^{\abs{\sigma}}u^j/\partial x^\sigma$. Then any
$\mathcal{C}$-field on~$\mathcal{E}$ is a linear combination of the
\emph{total derivatives}
\begin{equation}
  \label{eq:11}
  D_{x^i}=\pd{}{x^i}+\sum_{\sigma\in S,j\in J}u_{\sigma
    i}^j\pd{}{u_\sigma^j},\qquad
  i=1,\dots,n.
\end{equation}
The linearization of~$\mathcal{E}$ is the matrix operator with the entries
\begin{equation}
  \label{eq:12}
  (\ell_{\mathcal{E}})_j^l=\sum_{\sigma\in S}\pd{F^l}{u_\sigma^j}D_\sigma,
  \qquad j\in J, \quad l=1,\dots,r.
\end{equation}
A symmetry~$\phi=(\phi^1,\dots,\phi^m)$ enjoys the equation
\begin{equation}
  \label{eq:13}
  \sum_{\sigma\in S,j\in J}\pd{F^l}{u_\sigma^j}D_\sigma(\phi^j)=0,\qquad
  l=1,\dots,r,
\end{equation}
and the corresponding field is the \emph{evolutionary vector field}
\begin{equation}
  \label{eq:14}
  \re_\phi=\sum_{\sigma\in S,j\in J}D_\sigma(\phi^j)\pd{}{u_\sigma^j},
\end{equation}
while the bracket of symmetries~$\phi_1$, $\phi_2$ is given by
\begin{equation}
  \label{eq:15}
  \{\phi_1,\phi_2\}^j=\re_{\phi_1}(\phi_2^j)-\re_{\phi_2}(\phi_1^j),\qquad
  j=1,\dots,m.
\end{equation}
The operator adjoint to~\eqref{eq:12} is
\begin{equation}
  \label{eq:16}
  (\ell_{\mathcal{E}}^*)_l^j=
  \sum_{\sigma\in S}(-1)^{\abs\sigma}D_\sigma\circ\pd{F^l}{u_\sigma^j},
  \qquad j\in J, \quad l=1,\dots,r,
\end{equation}
and a cosymmetry~$\psi=(\psi^1,\dots,\psi^r)$ satisfies the equation
\begin{equation}
  \label{eq:17}
  \sum_{\sigma,l}(-1)^{\abs\sigma}D_\sigma(\pd{F^l}{u_\sigma^j}\psi^l)=0,
  \qquad j\in J.
\end{equation}

\paragraph{Step 2.}
\label{sec:step-2}

Here we look for closed $1$-forms and construct Abelian coverings associated
to them. A horizontal form~$\omega=\sum_{i}A_i\d x^i$ is closed if
\begin{equation}
  \label{eq:18}
  D_{x^\alpha}(A_\beta)=D_{x^\beta}(A_\alpha),\qquad
  1\le\alpha<\beta\le n.
\end{equation}
Such a form gives rise to a nonlocal variable~$w=w_\omega$ that satisfies the
equations
\begin{equation}
  \label{eq:19}
  \pd{w}{x^i}=A_i,\qquad i=1,\dots,n.
\end{equation}
These equations are compatible on~\eqref{eq:10} due to~\eqref{eq:18}. Recall
that for~$n=2$ closed $1$-forms coincide with conservation laws. The total
derivatives lifted to the covering equation~$\tilde{\mathcal{E}}$ are
\begin{equation}
  \label{eq:20}
  \tilde{D}_{x^i}=D_{x^i}+A_i\pd{}{w},\qquad i=1,\dots,n.
\end{equation}

\paragraph{Step 3.}
\label{sec:step-3}

At this step we compute a number of particular symmetries and cosymmetries
(using equations~\eqref{eq:13} and~\eqref{eq:17}, resp.). They are used to
construct canonical nonlocal variables on the $\ell^*$-covering
(\emph{nonlocal vectors}) and on the $\ell$-covering (\emph{nonlocal forms}),
resp., at Step~4. We also use them as seed elements in series generated by
recursion operators.

\paragraph{Step 4}
\label{sec:step-4}
consists of construction of the $\ell$- and $\ell^*$-coverings and
introduction of canonical nonlocal variables over them (see Step~3). The
$\ell$-covering is obtained by adding to~\eqref{eq:9} the system of equations
\begin{equation}
  \label{eq:21}
  \sum_{\sigma\in S,j\in J}\pd{F^l}{u_\sigma^j}\pd{q^j}{x^\sigma}=0,\qquad
  l=1,\dots,r,
\end{equation}
cf.~with Eq.~\eqref{eq:12}, while the $\ell^*$-covering is given by
\begin{equation}
  \label{eq:22}
  \sum_{\sigma,l}(-1)^{\abs\sigma}\pd{(\pd{F^l}{u_\sigma^j}p^l)}{x^\sigma}=0,
  \qquad j\in J,
\end{equation}
that comes from~\eqref{eq:16}.

If~$\phi$ is a symmetry of~$\mathcal{E}$ then one can introduce a covering
over~$\mathcal{L}^*(\mathcal{E})$ described by the system
\begin{equation}
  \label{eq:23}
  \pd{\bar{p}}{x^i}=\sum_{\sigma,l}\Delta_{\sigma,i}^l(\phi)p_\sigma^l,\qquad
  i=1,\dots,n,
\end{equation}
where~$\Delta_{\sigma,i}^l$ are $\mathcal{C}$-differential operators (see
Theorem~\ref{thm:correpondence}). In a similar way, to any cosymmetry~$\psi$
there corresponds a covering
\begin{equation}
  \label{eq:24}
  \pd{\bar{q}}{x^i}=\sum_{\sigma,j}\nabla_{\sigma,i}^j(\psi)q_\sigma^j\qquad
  i=1,\dots,n,
\end{equation}
$\nabla_{\sigma,i}^j$ being $\mathcal{C}$-differential operators as well.  We
omit here a general description of these operators and refer the reader to the
particular case of our interest exposed in Sections~\ref{sec:nonlocal-forms}
and~\ref{sec:nonlocal-vectors}.

\paragraph{Step 5.}
\label{sec:step-5}
We now use Theorems~\ref{thm:ell-shad} and~\ref{thm:ellstar-shad} to construct
recursion operators and Hamiltonian and symplectic
structures. Let~$\psi_1,\dots,\psi_s$ cosymmetries of~$\mathcal{E}$. Let us
consider the covering~$\widetilde{\mathcal{L}(\mathcal{E})}$
over~$\mathcal{L}(\mathcal{E})$ with the nonlocal
variables~$\bar{q}_1,\dots,\bar{q}_s$ defined by~\eqref{eq:24} and lift the
operators~$\ell_{\mathcal{E}}$ and~$\ell_{\mathcal{E}}^*$ to this
covering. Then the following result specifies Theorem~\ref{thm:ell-shad}:
\begin{theorem}\label{thm:nloc-ell}
  Let~$\Phi=(\Phi^1,\dots,\Phi^m)$ be a solution of the
  equation~$\tilde{\ell}_{\mathcal{E}}(\Phi)=0$
  on~$\widetilde{\mathcal{L}(\mathcal{E})}$ linear w.r.t.~$q_\sigma^\alpha$
  and $\bar{q}_\beta$\textup{:}
  \begin{equation*}
    \Phi^j=\sum_{\alpha,\sigma}a_\sigma^{\alpha,j}q_\sigma^\alpha
    +\sum_\beta b_\beta^j\bar{q}_\beta.
  \end{equation*}
  Then the operator
  \begin{equation*}
    \tens{R}=\sum_{\sigma}a_\sigma^{\alpha,j}D_\sigma+
    \sum_\beta b_\beta^jD_{x^i}^{-1}\circ
    \sum_{\sigma}\nabla_{\sigma,i}^\alpha(\psi_\beta)D_\sigma,
  \end{equation*}
  takes shadow of symmetries to shadows of symmetries. In a similar
  way\textup{,} to any solution~$\Psi=(\Psi^1,\dots,\Psi^r)$\textup{,}
  \begin{equation*}
    \Psi^j=\sum_{\alpha,\sigma}c_\sigma^{\alpha,j}q_\sigma^\alpha
    +\sum_\beta d_\beta^j\bar{q}_\beta
  \end{equation*}
  there corresponds the operator
  \begin{equation*}
    \tens{S}=\sum_{\sigma}c_\sigma^{\alpha,j}D_\sigma+
    \sum_\beta d_\beta^jD_{x^i}^{-1}\circ
    \sum_{\sigma}\nabla_{\sigma,i}^\alpha(\psi_\beta)D_\sigma
  \end{equation*}
  that takes shadows of symmetries to shadows of cosymmetries.
\end{theorem}

In a dual way, consider symmetries~$\phi_1,\dots,\phi_s$ of the
equation~$\mathcal{E}$ and the
covering~$\widetilde{\mathcal{L}^*(\mathcal{E})}$
over~$\mathcal{L}^*(\mathcal{E})$ with the nonlocal
variables~$\bar{p}_1,\dots,\bar{p}_s$ defined by~\eqref{eq:23}. Then,
lifting~$\ell_{\mathcal{E}}$ and~$\ell_{\mathcal{E}}^*$, we obtain a similar
specification of Theorem~\ref{thm:ellstar-shad}:

\begin{theorem}
  \label{thm:nloc-ellstar}
  Let~$\Phi=(\Phi^1,\dots,\Phi^m)$ be a solution of the
  equation~$\tilde{\ell}_{\mathcal{E}}(\Phi)=0$
  on~$\widetilde{\mathcal{L}^*(\mathcal{E})}$ linear
  w.r.t.~$p_\sigma^\alpha$ and $\bar{p}_\beta$\textup{:}
  \begin{equation*}
    \Phi^j=\sum_{\alpha,\sigma}a_\sigma^{\alpha,j}p_\sigma^\alpha
    +\sum_\beta b_\beta^j\bar{p}_\beta.
  \end{equation*}
  Then the operator
  \begin{equation*}
    \tens{H}=\sum_{\sigma}a_\sigma^{\alpha,j}D_\sigma+
    \sum_\beta b_\beta^jD_{x^i}^{-1}\circ
    \sum_{\sigma}\Delta_{\sigma,i}^\alpha(\phi_\beta)D_\sigma,
  \end{equation*}
  takes shadow of cosymmetries to shadows of symmetries. In a similar
  way\textup{,} to any solution~$\Psi=(\Psi^1,\dots,\Psi^r)$\textup{,}
  \begin{equation*}
    \Psi^j=\sum_{\alpha,\sigma}c_\sigma^{\alpha,j}p_\sigma^\alpha
    +\sum_\beta d_\beta^j\bar{p}_\beta
  \end{equation*}
  there corresponds the operator
  \begin{equation*}
    \hat{\tens{R}}=\sum_{\sigma}c_\sigma^{\alpha,j}D_\sigma+
    \sum_\beta d_\beta^jD_{x^i}^{-1}\circ
    \sum_{\sigma}\Delta_{\sigma,i}^\alpha(\phi_\beta)D_\sigma
  \end{equation*}
  that takes shadows of cosymmetries to shadows of cosymmetries.
\end{theorem}

After finding the operators~$\tens{R}$, $\tens{S}$, $\tens{H}$
and~$\hat{\tens{R}}$ we check conditions~\eqref{eq:6} and~\eqref{eq:7} and
compute necessary Schouten and Nijenhuis brackets.

\paragraph{Step 6}
\label{sec:step-6}
The last step consists of establishing algebraic relations between the
invariants constructed above.

\section{The matrix version}
\label{sec:matrix-version}

We consider Eq.~\eqref{eq:1} in the form
\begin{equation}
  \label{eq:26}
  u_t-u_{txx}+3uu_x=2u_xu_{xx}+uu_{xxx},
\end{equation}
i.e., set~$\mu=0$, and, similar
to~\eqref{eq:2}, introduce a new variable~$w=\alpha u-u_{xx}$, where~$\alpha$
is a new real constant. Consequently, the initial equation transforms to the system
\begin{equation}
  \label{eq:25}
  \begin{cases}
    w_t&=-2u_xw-uw_x,\\
    w&=\alpha u-u_{xx}.
  \end{cases}
\end{equation}
We choose the following variables for internal local coordinates on the
infinite prolongation of Eq.~\eqref{eq:25}:
\begin{equation*}
  x,\ t,\ w_l=\pd{^kw}{x^k},\ u_{0,k}=\pd{^ku}{t^k},\
  u_{1,k}=\pd{^{k+1}u}{x\,\partial t^k},
  \qquad k=0,1,\dots
\end{equation*}
Then the total derivatives in these coordinates will be of the form
\begin{align*}
  D_x=&\pd{}{x}+\sum_{k\ge0}w_{k+1}\pd{}{w_k}+\sum_{k\ge0}u_{1,k}\pd{}{u_{0,k}}+
  \sum_{k\ge0}D_t^k(\alpha u-w)\pd{}{u_{1,k}},\\
  D_t=&\pd{}{t}-\sum_{k\ge0}D_x^k(2u_{1,0}w+uw_1)\pd{}{w_k}+\sum_{k\ge0}u_{0,k+1}\pd{}{u_{0,k}}+
  \sum_{k\ge0}u_{1,k+1}\pd{}{u_{1,k}}.
\end{align*}

We introduce the following gradings:
\begin{equation*}
  \abs{x}=-1,\ \abs{t}=-2,\ \abs{u}=1,\ \abs{w}=3,\ \abs{\alpha}=2
\end{equation*}
and extend them in a natural way to all polynomial functions of the internal
coordinates. Then all computations can be restricted to homogeneous
components.

\subsection{Nonlocal variables}
\label{sec:nonlocal-variables-1}

In subsequent computations we shall need the following nonlocal variables
arising from conservation laws and defined by the equations
\begin{align*}
  (s_2)_x& = w,\\
  (s_2)_t& = ( - u^2\alpha - 2uw + u_1^2)/2;\\
  (s_3)_x& = uw,\\
  (s_3)_t& = - 2u^2w + uu_{1,1} - u_1u_{0,1};\\
  (s_6)_x& = w(uw - u_{1,1}),\\
  (s_6)_t& = ( - u^4\alpha^2 - 4u^3w\alpha + 2u^2u_1^2\alpha +
  4u^2u_1w_1 - 4u^2w^2 - 4uu_{0,2}\alpha + 12 uu_1^2w + 4uu_{1,1}w\\
  &- u_1^4 + 4u_1u_{1,2} - 4u_1u_{0,1}w)/4;\\
  (s_7)_x& = - 2u^3w_2 + 60u^2w^2 - 36uu_{1,1}w + 30 u_{0,2}w +
  27u_1^2u_{1,1},\\
  (s_7)_t& = - 104u^4w\alpha - 28u^4w_2 - 132u^3u_1w_1 + 36u^3
  u_{1,1}\alpha - 40u^3w^2 - 18u^2u_{0,2}\alpha - 48u^2u_1^2w \\
  & - 36u^2u_1u_{0,1}\alpha + 144u^2u_{1,1}w + 66 u^2u_{0,1}w_1 -
  78uu_{0,2}w - 36uu_1^2u_{1,1} + 18uu_1u_{1,2}\\
  & - 36uu_1u_{0,1}w + 30uu_{1,3} + 18uu_{1,1}^2 - 36uu_{0,1}^2\alpha + 9u_{0,2}u_1^2
  + 18u_{0,2}u_{1,1} - 30u_{0,3}u_1\\
  & + 36u_1^3u_{0,1} + 36u_1u_{1,1}u_{0,1} - 18u_{1,2} u_{0,1} + 18u_{0,1}^2w.
\end{align*}
The variable~$s_i$ is of grading~$i$ and computational experiment shows that
for every grading~$i=4n-2+\epsilon$, $\epsilon=0,1$, there exist an~$s_i$ such
that~$\abs{s_i}=i$.

In addition, we found conservation laws of fractional gradings:
\begin{align*}
  (s_{1/2})_x& = w^{1/2},\\
  (s_{1/2})_t& =  - w^{1/2}u;\\
  (s_{-1/2})_x& = w^{-3/2}(6w\alpha + w_2),\\
  (s_{-1/2})_t& = w^{-3/2}(4uw\alpha - uw_2 + 14w^2);\\
  (s_{-3/2})_x& = w^{-7/2}(12w^2\alpha^2 + 12ww_2\alpha - 2ww_4 + 7w_2^2),\\
  (s_{-3/2})_t& = w^{-7/2}(16uw^2\alpha^2 + 16uww_2\alpha + 2uww_4 - 7uw_2^2 -
  124w^3\alpha - 32w^2w_2);\\
  (s_{-5/2})_x& = w^{-11/2}(216w^3\alpha^3 + 540w^2w_2\alpha^2 - 180w^2w_4\alpha + 20w^2w_6 + 882ww_2^2\alpha\\
  & - 378ww_2w_4 - 16ww_3^2 + 837w_2^3),\\
  (s_{-5/2})_t & = w^{-11/2}(320uw^3\alpha^3 + 1008uw^2w_2\alpha^2 - 360uw^2w_4\alpha - 20uw^2w_6 \\
  &+ 1908uww_2^2\alpha + 378uww_2w_4 + 16uww_3^2 - 837uw_2^3 - 64u_1w^2w_3\alpha - 1400w^4\alpha^2 \\
  &- 1868w^3w_2\alpha + 580w^3w_4 + 64w^2w_1w_3 - 2322w^2w_2^2),
\end{align*}
etc.

\subsection{Symmetries}
\label{sec:symmetries-1}

A symmetry~$\phi=(\phi^w,\phi^u)$ of Eq.~\eqref{eq:25} must satisfy the
linearized equation
\begin{equation*}
  \begin{cases}
    D_t(\phi^w) + uD_x(\phi^w) + 2u_1\phi^w + 2wD_x(\phi^u) + w_1\phi^u =0,\\
    \phi^w + D_x^2(\phi^u) - \alpha\phi^u =0.
  \end{cases}
\end{equation*}
Direct computations lead to the following results.

\paragraph{(x,t)-independent symmetries.}
\label{sec:x-t-independent}

One can observe two types of symmetries that are independent of~$x$
and~$t$. The first one consists of symmetries of integer gradings:
\begin{align*}
  \phi^w_1& = w_1, \\
  \phi^u_1& = u_1;\\
  \phi^w_2& = uw_1 + 2u_1w, \\
  \phi^u_2& = - u_{0,1};\\
  \phi^w_5& = u^2w_1\alpha + 2uww_1 - u_1^2w_1 - 2u_{1,1}w_1 -
  4u_{0,1}w\alpha,\\
  \phi^u_5& = u^2u_1\alpha + 2u^2w_1 + 6uu_1w - u_1^3 + 2u_{1,2} -
  2u_{0,1}w;\\
  \phi^w_6& = u^3w_1\alpha + 2u^2u_1w\alpha + 8u^2ww_1 - uu_1^2w_1 +
  12uu_1w^2 - 2uu_{1,1} w_1 \\
  & + 2u_{0,2}w_1 - 2u_1^3w  +2u_1u_{0,1}w_1 + 4u_{1,2}w - 4u_{0,1}w^2,\\
  \phi^u_6& = 2u^3w_1 + 6u^2u_1w - 3u^2u_{0,1}\alpha - 6uu_{0,1}w - 2u_{0,3} +
  3u_1^2u_{0,1},
\end{align*}
etc. Symmetries of the second type have semi-integer gradings:
\begin{align*}
  \phi^w_{-3/2}& = ( - 4w^2w_1\alpha + 4w^2w_3 - 18ww_1w_2 +
  15w_1^3)/(2w^{7/2}),\\
  \phi^u_{-3/2}& =  - 2w_1/(w^{3/2});\\
  \phi^w_{-5/2}& = ( - 48w^4w_1\alpha^2 + 80w^4w_3\alpha -
  32w^4w_5 - 520w^3w_1w_2\alpha + 320w^3w_1w_4 \\
  & + 560w^3w_2w_3 + 560w^2w_1^3\alpha - 1820w^2w_1^2w_3 -
  2520w^2w_1w_2^2 + 6930ww_1^3w_2 \\
  & - 3465w_1^5)/(12w^{13/2}),\\
  \phi^u_{-5/2}& = ( - 12w^2w_1\alpha + 8w^2w_3 - 40ww_1w_2 +
  35w_1^3)/(3w^{9/2}),
\end{align*}
etc. All symmetries are local and~$\abs{\phi_\gamma}=\gamma$.

If one adds to the nonlocal setting the variables~$s_\gamma$ (see above) then
an additional series of nonlocal symmetries arises:
\begin{align*}
  \bar\phi^w_{-1}& = (s_{1/2}w^{-1/2}(4w^2w_1\alpha - 4w^2w_3 + 18ww_1w_2 - 15w_1^3) + 16tw^3w_1 \\
  &+ 2w( - 4w^2\alpha - 4ww_2 + 5w_1^2))/(16w^3),\\
  \bar\phi^u_{-1}& = (s_{1/2}w^{-1/2}w_1 + 4tu_1w - 2w)/(4w);\\
  \bar\phi^w_{-2}& = (3s_{1/2}w^{-1/2}( - 48w^4w_1\alpha^2 + 80w^4w_3\alpha - 32w^4w_5 - 520w^3w_1w_2\alpha + 320w^3w_1w_4 \\
  & + 560w^3w_2w_3 + 560w^2w_1^3\alpha - 1820w^2w_1^2w_3 - 2520w^2w_1w_2^2 + 6930ww_1^3w_2 - 3465w_1^5)\\
  & + 2s_{-1/2}w^{-1/2}w^3( - 4w^2w_1\alpha + 4w^2w_3 - 18ww_1w_2 + 15w_1^3) + 2w(96w^4\alpha^2  \\
  & + 384w^3w_2\alpha- 192w^3w_4 - 740w^2w_1^2\alpha + 1364w^2w_1w_3 +
  960w^2w_2^2 \\
  &- 5514ww_1
  ^2w_2 + 3465w_1^4))/(4w^6),\\
  \bar\phi^u_{-2}& = (3s_{1/2}w^{-1/2}( - 12w^2w_1\alpha + 8w^2w_3 - 40ww_1w_2 + 35w_1^3) - 2s_{-1/2}w^{-1/2}w^3w_1\\
  & + 2w(24w^2\alpha + 24ww_2 - 35w_1^2))/w^4,
\end{align*}
etc.

\paragraph{(x,t)-dependent symmetries.}
\label{sec:x-t-dependent}

The first three symmetries that depend on~$x$ and~$t$ are
\begin{align*}
  \phi^w_0& = t( - uw_1 - 2u_1w) + w,\\
  \phi^u_0& = tu_{0,1} + u;\\
  \phi^w_3& = t( - u^2w_1\alpha - 2uww_1 + u_1^2w_1 + 2u_{1,1}w_1 +
  4u_{0,1}w\alpha) + 2(s_2w_1 + 2uw\alpha + u_1 w_1),\\
  \phi^u_3& = t( - u^2u_1\alpha - 2u^2w_1 - 6uu_1w + u_1^3 -
  2u_{1,2} + 2u_{0,1}w) + 2(s_2u_1 + 2uw - 2 u_{1,1});\\
  \phi^w_4& = t( - u^3w_1\alpha - 2u^2u_1w\alpha - 8u^2ww_1 +
  uu_1^2w_1 - 12uu_1w^2 + 2uu_{1,1}w_1 - 2 u_{0,2}w_1 + 2u_1^3w \\
  & - 2u_1u_{0,1}w_1 - 4u_{1,2}w + 4u_{0,1}w^2) + 2(s_2(uw_1 +
  2u_1w) + s_3w_1 + 4u w^2 \\
  &- 4u_{1,1}w - 2u_{0,1}w_1),\\
  \phi^u_4& = t( - 2u^3w_1 - 6u^2u_1w + 3u^2u_{0,1}\alpha +
  6uu_{0,1}w + 2u_{0,3} - 3u_1^2u_{0,1}) + 2( - s_2u_{0,1} + s_3u_1\\
  &+ u^3\alpha + 3u^2w - uu_1^2 + 3u_{0,2}).
\end{align*}
All these symmetries, except for the first one, are nonlocal (description of
the nonlocal variable is given in Subsection~\ref{sec:nonlocal-variables-1})
and, as above, the subscript denotes the grading.

\subsection{Cosymmetries}
\label{sec:cosymmetries-1}

The defining equation for cosymmetries~$\psi=(\psi^w,\psi^u)$ is the adjoint
to the linearization of~\eqref{eq:25}:
\begin{align*}
  \begin{cases}
    D_t(\psi^w) + uD_x(\psi^w) - u_1\psi^w - \psi^u=0,\\
    2wD_x(\psi^w) + w_1\psi^w - D_x^2(\psi^u) + \alpha\psi^u=0.
  \end{cases}
\end{align*}
Similar to symmetries, we consider two types of cosymmetries.

\paragraph{(x,t)-independent cosymmetries.}
\label{sec:x-t-independent-1}

They are local and may be of integer and semi-integer gradings:
\begin{align*}
  \psi^w_3& = 1,\\
  \psi^u_3& = -u_1;\\
  \psi^w_4& = u,\\
  \psi^u_4& = u_{0,1};\\
  \psi^w_7& = u^2\alpha + 2uw - u_1^2 - 2u_{1,1},\\
  \psi^u_7& = - u^2u_1\alpha - 2u^2w_1 - 6uu_1w + u_1^3 - 2u_{1,2}
  + 2u_{0,1}w;\\
  \psi^w_8& = u^3\alpha + 4u^2w - uu_1^2 - 2uu_{1,1} + 2u_{0,2} +
  2u_1u_{0,1},\\
  \psi^u_8& = - 2u^3w_1 - 6u^2u_1w + 3u^2u_{0,1}\alpha + 6uu_{0,1}w + 2u_{0,3}
  - 3u_1^2u_{0,1}  \intertext{etc. and}
  \psi^w_{3/2}& = w^{-1/2},\\
  \psi^u_{3/2}& = 0;\\
  \psi^w_{1/2}& = (4w^2\alpha + 4ww_2 - 5w_1^2)/(4w^3w^{1/2}),\\
  \psi^u_{1/2}& = 2w_1/(ww^{1/2});\\
  \psi^w_{-1/2}& = (48w^4\alpha^2 + 160w^3w_2\alpha -
  64w^3w_4 - 280w^2w_1^2\alpha + 448w^2w_1w_3 + 336 w^2w_2^2 \\
  & - 1848ww_1^2w_2 + 1155w_1^4)/(48w^6w^{1/2}),\\
  \psi^u_{-1/2}& = (12w^2w_1\alpha - 8w^2w_3 + 40ww_1w_2 -
  35w_1^3)/(3w^4w^{1/2}),
\end{align*}
etc.

Similar to the case of symmetries, when one adds nonlocal variables~$s_\gamma$
an additional series of nonlocal cosymmetries arises:
\begin{align*}
  \bar\psi^w_1& = (3s_{1/2}w^{-1/2}( - 4w^2\alpha - 4ww_2 + 5w_1^2) - 2s_{-1/2}w^{-1/2}w^3 + 96tw^3 - 10ww_1)/(96w^3),\\
  \bar\psi^u_1& = ( - s_{1/2}w^{-1/2}w_1 - 4tu_1w + 2w)/(4w);\\
  \bar\psi^w_0& = (3s_{1/2}w^{-1/2}(48w^4\alpha^2 + 160w^3w_2\alpha - 64w^3w_4 - 280w^2w_1^2\alpha + 448w^2w_1w_3 \\
  & + 336w^2w_2^2- 1848ww_1^2w_2 + 1155w_1^4) + 4s_{-1/2}w^{-1/2}w^3(4w^2\alpha + 4ww_2 - 5w_1^2)  \\
  &+ 12s_{-3/2}w^{-1/2}w^6+ 2w(400w^2w_1\alpha - 276w^2w_3 + 1366ww_1w_2 - 1155w_1^3))/(16w^6),\\
  \bar\psi^u_0& = (3s_{1/2}w^{-1/2}(12w^2w_1\alpha - 8w^2w_3 + 40ww_1w_2 - 35w_1^3) + 2s_{-1/2}w^{-1/2}w^3w_1 \\
  &+ 2w( - 24w^2\alpha - 24ww_2 + 35w_1^2))/w^4;\\
  \bar\psi^w_{-1}& = (27s_{1/2}w^{-1/2}(320w^6\alpha^3 + 2240w^5w_2\alpha^2 - 1792w^5w_4\alpha + 512w^5w_6 - 5040w^4w_1^2\alpha^2 \\
  &+ 16128w^4w_1w_3\alpha - 6912w^4w_1w_5 + 12096w^4
  w_2^2\alpha - 14592w^4w_2w_4 - 8832w^4w_3^2\\
  & - 81312w^3w_1^2w_2\alpha + 52800w^3w_1^2w_4 + 181632w^3w_1w_2w_3 + 42944w^3w_2^3  \\
  & + 60060w^2w_1^4\alpha- 274560w^2w_1^3w_3 - 569712w^2w_1^2w_2^2 + 1021020ww_1^4w_2 - 425425w_1^6) \\
  &+ 18s_{-1/2}w^{-1/2}w^3(48w^4\alpha^2 + 160w^3w_2\alpha - 64w^3w_4 - 280w^2w_1^2\alpha + 448w^2w_1w_3 \\
  & + 108s_{-3/2}w^{-1/2}w^6(4w^2\alpha + 4ww_2 - 5w_1^2) + 40s_{-5/2}w^{-1/2}w^9 + 2w(69120w^4w_1\alpha^2 \\
  & + 336w^2w_2^2 - 1848ww_1^2w_2 + 1155w_1^4)+ 34160w^4w_5 + 624816w^3w_1w_2\alpha \\
  & - 349880w^3w_1w_4 - 605776w^3w_2w_3 - 665280w^2w_1^3\alpha +
  1991880w^2w_1^2w_3\\
  &+ 2772036w^2w_1w_2^2 - 7636860ww_1^3w_2 + 3828825w_1^5))/(72w^9),\\
  \bar\psi^u_{-1}& = (3s_{1/2}w^{-1/2}(240w^4w_1\alpha^2 - 320w^4w_3\alpha + 128w^4w_5 + 2240w^3w_1w_2\alpha \\
  & - 1344w^3w_1w_4 - 2240w^3w_2w_3 - 2520w^2w_1^3\alpha + 7728w^2w_1^2
  w_3 + 10416w^2w_1w_2^2 \\
  &- 29568ww_1^3w_2 + 15015w_1^5) + 4s_{-1/2}w^{-1/2}w^3(12w^2w_1\alpha - 8w^2w_3 + 40ww_1w_2  \\
  & - 35w_1^3)+ 12s_{-3/2}w^{-1/2}w^6w_1 + 2w( -384w^4\alpha^2 -
  1536w^3w_2\alpha +
  768w^3w_4 \\
  & + 3360w^2w_1^2\alpha- 5732w^2w_1w_3- 3840w^2w_2^2 + 23422ww_1^2w_2 -
  15015w_1^4))/w^7,
\end{align*}
etc.

\paragraph{(x,t)-dependent cosymmetries.}
\label{sec:x-t-dependent-1}

All them are nonlocal:
\begin{align*}
  \psi^w_5& = t( - u^2\alpha - 2uw + u_1^2 + 2u_{1,1}) + 2(s_2 +
  u_1),\\
  \psi^u_5& = t(u^2u_1\alpha + 2u^2w_1 + 6uu_1w - u_1^3 + 2u_{1,2} -
  2u_{0,1}w) + 2( - s_2u_1 - 2uw + 2u_{1,1});\\
  \psi^w_6& = t( - u^3\alpha - 4u^2w + uu_1^2 + 2uu_{1,1} - 2u_{0,2}
  - 2u_1u_{0,1}) + 2(s_2u + s_3 - 2u_{0,1}),\\
  \psi^u_6& = t(2u^3w_1 + 6u^2u_1w - 3u^2u_{0,1}\alpha - 6uu_{0,1}w
  - 2u_{0,3} + 3u_1^2u_{0,1}) + 2(s_2u_{0,1} - s_3u_1 \\
  &- u^3\alpha - 3u^2w + uu_1^2 - 3u_{0,2}),
\end{align*}
etc.

\subsection{Nonlocal forms}
\label{sec:nonlocal-forms-1}

Recall that nonlocal forms are nonlocal variables of a special type on the
$\ell$-covering. The $\ell$-covering itself is obtained from Eq.~\eqref{eq:25}
by adding two additional equations
\begin{equation*}
  \begin{cases}
    q^w_t&= - uq^w_x - 2u_xq^w - 2wq^u_x - w_xq^u,\\
    q^u_{xx}&= \alpha q^u - q^w,
  \end{cases}
\end{equation*}
where~$q^w$ and~$q^u$ are new odd variables. The total derivatives on the
$\ell$-covering are
\begin{align*}
  \tilde{D}_x&=D_x+\sum_{k\ge0}q^w_{k+1}\pd{}{q^w_k}+ \sum_{k\geq
    0}q^u_{1,k}\pd{}{q^u_{0,k}}+
  \sum_{k\ge0}\tilde{D}_t^{k}(\alpha q^u-q^w)\pd{}{q^u_{1,k}},\\
  \tilde{D}_t&=D_t-\sum_{k\ge0}\tilde{D}^k_x(uq^w_1+2u_1q^w+2wq^u_1+
  w_1q^u)\pd{}{q^w_k}+ \sum_{k\ge0}q^u_{0,k+1}\pd{}{q^u_{0,k}}+
  \sum_{k\ge0}q^u_{1,k+1}\pd{}{q^u_{1,k}}.
\end{align*}
The nonlocal form~$Q_i$ associated to a
cosymmetry~$\psi_i=(\psi_i^w,\psi_i^u)$ (see
Subsection~\ref{sec:cosymmetries-1}) is defined by the equations
\begin{align*}
  \tilde{D}_x(Q^i)& =\psi_i^wq^w,\\
  \tilde{D}_t(Q^i)& =-u\psi_i^wq^w +(D_x(\psi_i^u)-2w\psi_i^w)q^u
  -\psi_i^uq_1^u.
\end{align*}

\subsection{Nonlocal vectors}
\label{sec:nonlocal-vectors-1}

Dually to nonlocal forms, nonlocal vectors arise as special nonlocal variables
on the $\ell^*$-covering associated to symmetries of the initial equation. The
$\ell^*$-covering is the extension of Eq.~\eqref{eq:25} by two new equations
\begin{equation*}
  \begin{cases}
    p^w_t= - up^w_x + u_xp^w + p^u,\\
    p^u_{xx}= 2wp^w_x + w_xp^w + \alpha p^u,
  \end{cases}
\end{equation*}
where~$p=(p^w,p^u)$ is a new odd variable. The total deriavatives are given by
\begin{align*}
  \tilde{D}_x&=D_x+\sum_{k\ge0}p_{k+1}^w\pd{}{p^w_k} +
  \sum_{k\ge0}p_{1,k}^u\pd{}{p^u_{0,k}}+\sum_{k\ge0}\tilde{D}_t^k(2wp_1^w+w_1p^w+\alpha p^u)\pd{}{p_{1,k}^u},\\
  \tilde{D}_t&=D_t+\sum_{k\ge0}\tilde{D}_x^k(-up_1^w+u_1p^w+p^u)\pd{}{p_k^w}+
  \sum_{k\ge0}p_{0,k+1}^u\pd{}{p_{0,k}^u}+\sum_{k\ge0}p_{1,k+1}^u\pd{}{p_{1,k}^u}.
\end{align*}
The nonlocal vector~$P_i$ associated to a symmetry~$\phi=(\phi^w,\phi^u)$ (see
Subsection~\ref{sec:symmetries-1}) is defined by the equations
\begin{align*}
  \tilde{D}_x(P^i)&=\phi_i^wp^w,\\
  \tilde{D}_t(P^i)&=-(u\phi_i^w+2w\phi_i^u)p^w-D_x(\phi_i^u)p^u+\phi_i^up^u_1.
\end{align*}

\subsection{Recursion operators for symmetries}
\label{sec:recurs-oper-symm-1}

The defining equations for these operators are
\begin{equation*}
  \begin{cases}
    \tilde{D}_t(\tens{R}^w) + u\tilde{D}_x(\tens{R}^w) + 2u_1\tens{R}^w +
    2w\tilde{D}_x(\tens{R}^u) + w_1\tens{R}^u =0,\\
    \tens{R}^w + \tilde{D}_x^2(\tens{R}^u) - \alpha\tens{R}^u =0
  \end{cases}
\end{equation*}
(see Theorem~\ref{thm:nloc-ell}), where the total derivatives are those
described in Subsection~\ref{sec:nonlocal-forms-1}. The following two solutions
are essential:
\begin{align*}
  \tens{R}_{-1}^w& = ( Q^{3/2}w^{-1/2}(- 4w^2w_1\alpha + 4w^2w_3 -
  18ww_1w_2 + 15w_1^3) \\
  & - 4q_2^ww^2 + 10q_1^www_1 + q^w(4w^2\alpha + 8ww_2 - 15w_1^2))/ w^3,\\
  \tens{R}_{-1}^u& = 4( - Q^{3/2}w^{-1/2}w_1 + q^w)/w \intertext{and}
  \tens{R}_3^w& = Q^3w_1 + q_1^uw_1 + 2q^uw\alpha,\\ 
  \tens{R}_3^u& = Q^3u_1 + q^wu - q_{1,1}^u + q^uw.
\end{align*}
The corresponding operators are of the form
\begin{align*}
  \tens{R}_{-1}&=\frac{1}{8w^3}
  \begin{pmatrix}
    h_{12}w^{-1/2}D_x^{-1}\circ w^{-1/2}/2 - 4w^2D_x^2 + 10ww_1D_x + h_8&\ &0\\
    -4w^2w_1w^{-1/2}D_x^{-1}\circ w^{-1/2}/2 + 4w^2&\ &0
  \end{pmatrix}
  \intertext{and}
  \tens{R}_3&=
  \begin{pmatrix}
    w_1D_x^{-1}     &\ & w_1D_x + 2w\alpha  \\
    u_1D_x^{-1} + u &\ & -D_{xt} + w \\
  \end{pmatrix},
\end{align*}
where~$h_{12}= - 4w^2w_1\alpha + 4w^2w_3 - 18ww_1w_2 + 15w_1^3$,
$h_8=4w^2\alpha + 8ww_2 - 15w_1^2$. All other solutions obtained in our
computations corresponded to operators that are generated by the two above.

\subsection{Symplectic structures}
\label{sec:sympl-struct-1}

Symplectic structures, as it follows from Theorem~\ref{thm:nloc-ell}, are
defined by the equations
\begin{equation*}
  \begin{cases}
    -\tilde{D}_t(\tens{S}^w) - u\tilde{D}_x(\tens{S}^w) + u_1\tens{S}^w + \tens{S}^u=0,\\
    -2w\tilde{D}_x(\tens{S}^w) - w_1\tens{S}^w + \tilde{D}_x^2(\tens{S}^u) -
    \alpha\tens{S}^u=0,
  \end{cases}
\end{equation*}
where the total derivatives were defined in
Subsection~\ref{sec:nonlocal-forms-1}. Here are the simplest nontrivial
solutions:
\begin{align*}
  \tens{S}_2^w& = Q^{3/2}w^{-1/2},\\
  \tens{S}_2^u& = - q^u;\\
  \tens{S}_5^w& = Q^3 + q^u_1,\\
  \tens{S}_5^u& =  - Q^3u_1 - q^wu + q^u_{1,1} - q^uw;\\
  \tens{S}_6^w& = Q^4 + Q^3u - q^u_{0,1} + q^u_1u - q^uu_1,\\
  \tens{S}_6^u&= - Q^4u_1 + Q^3u_{0,1} - q^wu^2 - q^u_{0,2} + q^u( - u^2\alpha -
  2uw + u_1^2)
\end{align*}
with the corresponding symplectic operators
\begin{align*}
  \tens{S}_2&=
  \begin{pmatrix}
    w^{-1/2}D_x^{-1}\circ w^{-1/2}/2 &\ & 0 \\
    0 &\ & - 1
  \end{pmatrix},\\
  \tens{S}_5&=
  \begin{pmatrix}
    D_x^{-1}          &\ & D_x \\
    - u_1 D_x^{-1} - u &\ & D_{xt} - w
  \end{pmatrix},\\
  \tens{S}_6&=
  \begin{pmatrix}
    D_x^{-1}\circ u + uD_x^{-1}                  &\ & - D_t + uD_x - u_1 \\
    - u_1D_x^{-1}\circ u + u_{0,1}D_x^{-1} - u^2 &\ & - D_t^2 - u^2\alpha - 2uw +
    u_1^2
  \end{pmatrix}.
\end{align*}

\subsection{Hamiltonian structures}
\label{sec:hamilt-struct-1}

The equations that should be satisfied by a Hamiltonian operator (see
Theorem~\ref{thm:nloc-ellstar}) are
\begin{equation*}
  \begin{cases}
    \tilde{D}_t(\tens{H}^w) + u\tilde{D}_x(\tens{H}^w) + 2u_1\tens{H}^w +
    2w\tilde D_x(\tens{H}^u)
    + w_1\tens{H}^u =0,\\
    \tens{H}^w + \tilde{D}_x^2(\tens{H}^u) - \alpha\tens{H}^u =0,
  \end{cases}
\end{equation*}
where the total derivatives are from
Subsection~\ref{sec:nonlocal-vectors-1}. In particular, we found the following
solutions:
\begin{align*}
  \tens{H}_{-3}^w& =  - p_3^w + p_1^w\alpha,\\
  \tens{H}_{-3}^u& = p_1^w;\\
  \tens{H}_{-2}^w& = 2p_1^ww + p^ww_1,\\
  \tens{H}_{-2}^u& =  - p^u;\\
  \tens{H}_1^w& = P^1w_1 - 2p^www_1 + p_1^uw_1 + 2p^uw\alpha ,\\
  \tens{H}_1^u& = P^1u_1 - 2p_1^wuw - p^w(uw_1 + 2u_1w) - p_{1,1}^u+p^uw.
\end{align*}
The corresponding Hamiltonian operators are
\begin{align*}
  \tens{H}_{-3}&=
  \begin{pmatrix}
    - D_x^3 + \alpha D_x &\ & 0 \\
    D_x &\ & 0
  \end{pmatrix},\\
  \tens{H}_{-2}&=
  \begin{pmatrix}
    2w D_x + w_1 &\ & 0 \\
    0 &\ & - 1
  \end{pmatrix},\\
  \tens{H}_1&=
  \begin{pmatrix}
    w_1 D_x^{-1}\circ w_1 - 2ww_1                    &\ & w_1 D_x + 2w\alpha \\
    u_1 D_x^{-1}\circ w_1 - 2uw D_x - (uw_1 + 2u_1w) &\ & - D_{xt} + w
  \end{pmatrix}.
\end{align*}

\subsection{Recursion operators for cosymmetries}
\label{sec:recurs-oper-cosymm-1}

By Theorem~\ref{thm:nloc-ellstar}, the equation to find recursion operators
for cosymmetries are
\begin{align*}
  \begin{cases}
    -\tilde D_t(\hat{\tens{R}}^w) - u\tilde D_x(\hat{\tens{R}}^w) +
    u_1\hat{\tens{R}}^w
    + \hat{\tens{R}}^u=0,\\
    -2w\tilde D_x(\hat{\tens{R}}^w) - w_1\hat{\tens{R}}^w + \tilde
    D_x^2(\hat{\tens{R}}^u) - \alpha\hat{\tens{R}}^u=0
  \end{cases}
\end{align*}
with the total derivatives given in
Subsection~\ref{sec:nonlocal-vectors-1}. One of solutions is presented below:
\begin{align*}
  \hat{\tens{R}}_3^w& =  P^1 - 2p^ww + p^u_1,\\
  \hat{\tens{R}}_3^u& = - P^1u_1 +2p_1^wuw + p^w(uw_1 + 2u_1w) + p_{1,1}^u - p^uw
\end{align*}
The corresponding recursion operator is
\begin{equation*}
  \hat{\tens{R}}_3=
  \begin{pmatrix}
    D_x^{-1}\circ w_1 - 2w &\ & D_x \\
    -u_1D_x^{-1}\circ w_1 +2uwD_x + uw_1 + 2u_1w &\ & D_{xt} - w
  \end{pmatrix}.
\end{equation*}

\subsection{Interrelation}
\label{sec:interrelation}

We expose here basic facts on structural relations between the above described
invariants. The main one is the following
\begin{theorem}
  \label{thm:P-N}
  Recursion operator~$\tens{R}_3$ and Hamiltonian operator~$\tens{H}_{-3}$
  constitute a Poisson--Nijenhuis structure on the Camassa--Holm
  equation. Consequently\textup{,} all the
  operators~$\tens{R}_3^n\circ\tens{H}_{-3}$ are Hamiltonian and pair-wise
  compatible. In particular\textup{,}
  $\tens{R}_3\circ\tens{H}_{-3}=\alpha\tens{H}_{-2}$.
\end{theorem}
\begin{proof}
  The proof consists of direct computations using the results and techniques of
  Ref.~\cite{G-K-V-TMPh}\footnote{We are pretty sure that the
    pair~$(\tens{H}_{-3},\tens{R}_{-1})$ is also a Poisson--Nijenhuis
    structure and generates an infinite negative hierarchy of Hamiltonian
    operators, but could not prove this fact because of computer capacity
    limitations.}.
\end{proof}

A visual presentation of how symmetries are distributed over gradings is given
in Table~\ref{tab:2}.
\begin{table}[h]
\caption{Distribution of symmetries over gradings}
\label{tab:2}       
\begin{tabular}{l|ccccccccccccccccc}
\hline\noalign{\smallskip}
Gradings           &$-\frac{5}{2}$ & $-2$   & $-\frac{3}{2}$ & $-1$& $-\frac{1}{2}$& $0$&$1$&$2$&$3$&$4$&$5$&$6$&$7$&$8$\\
\noalign{\smallskip}\hline\noalign{\smallskip}
Local positive     & & & & & & &$\bullet$& $\bullet$ &&&$\bullet$&$\bullet$\\
Local negative     & $\bullet$& &$\bullet$ && & & & & \\
Nonlocal positive  & & & & & &$\bullet$ & & & $\bullet$& $\bullet$& & & $\bullet$ & $\bullet$\\
Nonlocal negative  & & $\bullet$& & $\bullet$& & & & & \\
\noalign{\smallskip}\hline
\end{tabular}
\end{table}
How to prove locality of the first two series of symmetries will be discussed
in Section~\ref{sec:conclusion}. Similar presentation for cosymmetries see in
Table~\ref{tab:3}.
\begin{table}[h]
\caption{Distribution of cosymmetries over gradings}
\label{tab:3}       
\begin{tabular}{l|ccccccccccccccccc}
\hline\noalign{\smallskip}
Gradings           &$-\frac{3}{2}$ & $-1$   & $-\frac{1}{2}$ & $0$& $\frac{1}{2}$& $1$&$\frac{3}{2}$&$2$&$3$&$4$&$5$&$6$&$7$&$8$\\
\noalign{\smallskip}\hline\noalign{\smallskip}
Local positive     & & & & & & && & $\bullet$& $\bullet$&&&$\bullet$&$\bullet$\\
Local negative     & $\bullet$& &$\bullet$ && $\bullet$ & &$\bullet$  & & \\
Nonlocal positive  & & & & & & & & $\bullet$ &&& $\bullet$ & $\bullet$ & &\\
Nonlocal negative  & & $\bullet$& & $\bullet$& & $\bullet$ & & & \\
\noalign{\smallskip}\hline
\end{tabular}
\end{table}

The action of Hamiltonian and recursion operators for symmetries (up to a
constant multiplier) is given in Diagram~\eqref{eq:28}:
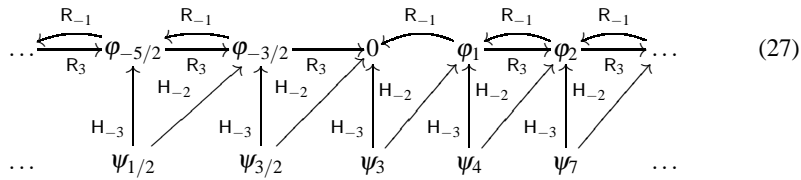
\begin{equation}\xymatrixrowsep{2.5pc}
  \label{eq:28}
  \xymatrix@1{
    \dots\ar[r]_{\tens{R}_3}&\phi_{-5/2}\ar[r]_{\tens{R}_3}\ar@/{_}/[l]_{\tens{R}_{-1}}&
    \phi_{-3/2}\ar[r]_{\tens{R}_3}\ar@/{_}/[l]_{\tens{R}_{-1}}&0&
    \phi_1\ar[r]_{\tens{R}_3}\ar@/{_}/[l]_{\tens{R}_{-1}}&
    \phi_2\ar[r]_{\tens{R}_3}\ar@/{_}/[l]_{\tens{R}_{-1}}
    &\dots\ar@/{_}/[l]_{\tens{R}_{-1}}\\
    \dots&\psi_{1/2}\ar[ur]^{\tens{H}_{-2}}\ar[u]^(.3){\tens{H}_{-3}}&
    \psi_{3/2}\ar[ur]^{\tens{H}_{-2}}\ar[u]^(.3){\tens{H}_{-3}}&
    \psi_3\ar[ur]^{\tens{H}_{-2}}\ar[u]^(.3){\tens{H}_{-3}}&
    \psi_4\ar[ur]^{\tens{H}_{-2}}\ar[u]^(.3){\tens{H}_{-3}}&
    \psi_7\ar[ur]^{\tens{H}_{-2}}\ar[u]^(.3){\tens{H}_{-3}}
    &\dots
  }
\end{equation}
The action of recursion operators for cosymmetries and simplectic structures
is similar.

We shall now prove commutativity of the local hierarchies.

\begin{lemma}
  \label{lem:action}
  The symmetry~$\bar{\phi}_3$ is a positive hereditary symmetry\textup{,}
  i.e.\textup{,} its action on local symmetries\textup{,}
  $\phi\mapsto\{\bar{\phi}_3,\phi\}$\textup{,} coincides\textup{,} up to a
  multiplier\textup{,} with the one of the recursion
  operator~$\tens{R}_3$. The only symmetries that vanish under this action
  are~$\phi_{-3/2}$ and~$\phi_1$. In a similar way\textup{,} the
  symmetry~$\phi_{-1}$ is a negative hereditary symmetry and the only symmetry
  that is taken to zero under its action is~$\phi_1$.
\end{lemma}

A direct corollary of this result is

\begin{theorem}
  \label{thm:commut}
  Local positive and negative symmetries form commutative hierarchies.
\end{theorem}

\section{The scalar version}
\label{sec:scalar-version}

Let us consider now the Camassa--Holm equation in its initial
form~\eqref{eq:1} with~$\mu=0$ and, similar to the matrix case introduce a new
real parameter~$\alpha$:
\begin{equation}
  \label{eq:27}
  \alpha u_t - u_{txx} + 3\alpha uu_x = 2u_xu_{xx} + uu_{xxx}.
\end{equation}
For the internal coordinates we choose the functions
\begin{equation*}
  u_l=\pd{^lu}{x^l},\quad u_{l,k}=\pd{^{k+l}u}{x^l\partial t^k},\qquad
  l=0,1,2,\ k\ge 1.
\end{equation*}
The total derivatives in these coordinates are of the form
\begin{align*}
  D_x&=\pd{}{x}+\sum_{l=0}^2u_{l+1}\pd{}{u_l}+
  \sum_{k\ge1}\left(u_{1,k}\pd{}{u_{0,k}}+u_{2,k}\pd{}{u_{1,k}}+
    D_t^{k}(u_3)\pd{}{u_{2,k}}\right),\\
  D_t&=\pd{}{t}+\sum_{l=0}^2u_{l,1}\pd{}{u_l}+
  \sum_{l=0}^2\sum_{k\ge1}u_{l,k+1}\pd{}{u_{l,k}},
\end{align*}
where~$u_3=(\alpha u_{0,1} - u_{2,1} + 3\alpha uu_1 - 2u_1u_2)/u$. The
equation becomes homogeneous if assign the following gradings:
\begin{equation*}
  \abs{x}= -1,\ \abs{t}= -2,\ \abs{u}= 1,\ \abs{\alpha}=2.
\end{equation*}

\subsection{Nonlocal variables}
\label{sec:nonlocal-variables}

We introduce nonlocal variable associated to conservation laws of the equation
at hand:
\begin{align*}
  (s_2)_x& = u\alpha - u_2,\\
  (s_2)_t& = ( - 3u^2\alpha + 2uu_2 + u_1^2)/2;\\
  (s_3)_x& = u(u\alpha - u_2),\\
  (s_3)_t& =  - 2u^3\alpha + 2u^2u_2 + uu_{1,1} - u_{0,1}u_1;\\
  (s_6)_x& = u^3\alpha^2 - 2u^2u_2\alpha - uu_{1,1}\alpha + uu_2^2 +
  u_{1,1}u_2,\\
  (s_6)_t& = ( - 9u^4\alpha^2 + 12u^3u_2\alpha + 6u^2u_1^2\alpha +
  4u^2u_{1,1}\alpha - 4u^2u_2^2 - 8u u_{0,1}u_1\alpha\\
  & - 4uu_{0,2}\alpha - 4uu_1^2u_2 + 4uu_1u_{2,1} - 4uu_{1,1}u_2 +
  4u_{0,1}u_1u_2 - u_1^4 + 4u_1 u_{1,2})/4;\\
  (s_7)_x& = 60u^4\alpha^2 - 116u^3u_2\alpha - 12u^2u_1^2\alpha -
  40u^2u_{1,1}\alpha + 56u^2u_2^2 - 6u u_{0,1}u_1\alpha \\
  & + 28uu_{0,2}\alpha + 12uu_1^2u_2 + 10uu_1u_{2,1} + 40uu_{1,1}u_2
  + 2uu_{2,2} + 2u_{0,1}^2\alpha - 4 u_{0,1}u_1u_2 \\
  & - 2u_{0,1}u_{2,1} - 30u_{0,2}u_2 + 27u_1^2u_{1,1},\\
  (s_7)_t& = - 144u^5\alpha^2 + 240u^4u_2\alpha + 48u^3u_1^2\alpha
  + 124u^3u_{1,1}\alpha - 96u^3u_2^2 \\
  & - 156 u^2u_{0,1}u_1\alpha - 124u^2u_{0,2}\alpha - 48u^2u_1^2u_2
  + 8u^2u_1u_{2,1} - 88u^2u_{1,1}u_2 + 28u^2 u_{2,2} \\
  & - 56uu_{0,1}^2\alpha + 112uu_{0,1}u_1u_2 + 38uu_{0,1}u_{2,1} +
  78uu_{0,2}u_2 - 36uu_1^2u_{1,1} + 18 uu_1u_{1,2}\\
  & + 18uu_{1,1}^2 + 30uu_{1,3} - 18u_{0,1}^2u_2 + 36u_{0,1}u_1^3 +
  36u_{0,1}u_1u_{1,1} - 18 u_{0,1}u_{1,2} + 9u_{0,2}u_1^2 \\
  & + 18u_{0,2}u_{1,1} - 30u_{0,3}u_1.
\end{align*}

\subsection{Symmetries}
\label{sec:symmetries}

A symmetry~$\phi$ must satisfy the linearized equation
\begin{equation*}
  \alpha D_t(\phi) -D^2_xD_t(\phi) - uD_x^3(\phi) - 2u_1D_x^2(\phi)
  + (3\alpha u - 2u_2)D_x(\phi) + (3\alpha u_1 - u_3)\phi=0.
\end{equation*}

We computed two types of symmetries. Everywhere below the subscript was chosen
in a way to correspond the enumeration taken for the matrix case.

\paragraph{$(x,t)$-independent symmetries.}
\label{sec:x-t-independent-2}

We present the first four of them:
\begin{align*}
  \phi_{1}& = u_1,\\
  \phi_{2}& =  - u_{0,1},\\
  \phi_{5}& = 3u^2u_1\alpha - 4uu_{0,1}\alpha - 2uu_1u_2 + 2uu_{2,1}
  + 2u_{0,1}u_2 - u_1^3 + 2u_{1,2},\\
  \phi_{6}& = 2u^3u_1\alpha - 11u^2u_{0,1}\alpha - 2u^2u_1u_2 + 2u^2u_{2,1} +
  6uu_{0,1}u_2 + 3u_{0,1}u_1^2 - 2u_{0,3}.
\end{align*}
All these symmetries are local.

\paragraph{$(x,t)$-dependent symmetries.}
\label{sec:x-t-dependent-2}

The first three $(x,t)$-dependent symmetries are
\begin{align*}
  \phi_{0}& = tu_{0,1} + u,\\
  \phi_{3}& = 2s_2u_1 + t( - 3u^2u_1\alpha + 4uu_{0,1}\alpha +
  2uu_1u_2 - 2uu_{2,1} - 2u_{0,1}u_2 + u_1^3 - 2 u_{1,2}) \\
  & + 4(u^2\alpha - uu_2 - u_{1,1}),\\
  \phi_{4}& = - 2s_2u_{0,1} + 2s_3u_1 + t( - 2u^3u_1\alpha +
  11u^2u_{0,1}\alpha + 2u^2u_1u_2 - 2u^2 u_{2,1} \\
  & - 6uu_{0,1}u_2 - 3u_{0,1}u_1^2 + 2u_{0,3}) + 2(4u^3\alpha - 3u^2u_2 -
  uu_1^2 + 3u_{0,2})
\end{align*}
The only local symmetry in this series is~$\phi_1$.

\subsection{Cosymmetries}
\label{sec:cosymmetries}

The defining equation for cosymmetries is
\begin{equation*}
  -\alpha D_t(\psi) + D_x^2D_t(\psi) + uD_x^3(\psi) + u_1D_x^2(\psi)
  + (u_2 - 3\alpha u)D_x(\psi) =0.
\end{equation*}

\paragraph{$(x,t)$-independent cosymmetries.}
\label{sec:x-t-independent-3}

These cosymmetries are local. The first four of them are
\begin{align*}
  \psi_{3}& = 1,\\
  \psi_{4}& = u,\\
  \psi_{7}& = 3u^2\alpha - 2uu_2 - u_1^2 - 2u_{1,1},\\
  \psi_{8}& = 5u^3\alpha - 4u^2u_2 - uu_1^2 - 2uu_{1,1} + 2u_{0,1}u_1 +
  2u_{0,2}.
\end{align*}

\paragraph{$(x,t)$-dependent cosymmetries.}
\label{sec:x-t-dependent-3}

These cosymmetries are nonlocal, except for the first one:
\begin{align*}
  \bar\psi_2 &= -tu,\\
  \bar\psi_{5}& = 2s_2 + t( - 3u^2\alpha + 2uu_2 + u_1^2 + 2u_{1,1})
  + 2u_1,\\
  \bar\psi_{6}& = 2s_2u + 2s_3 + t( - 5u^3\alpha + 4u^2u_2 + uu_1^2 +
  2uu_{1,1} - 2u_{0,1}u_1 - 2u_{0,2}) - 4 u_{0,1},
\end{align*}
etc.

\subsection{Nonlocal forms}
\label{sec:nonlocal-forms}

Nonlocal forms arise in the $\ell$-covering which is given by the equation
\begin{equation*}
  \alpha q_t - q_{xxt} - uq_{xxx} - 2u_xq_{xx}
  + (3\alpha u - 2u_{xx})q_x + (3\alpha u_x - u_{xxx})q=0
\end{equation*}
with the total derivatives
\begin{align*}
  \tilde{D}_x&=D_x+\sum_{l=0}^2q_{l+1}\pd{}{q_l}+
  \sum_{k\ge1}\left(q_{1,k}\pd{}{q_{0,k}} + q_{2,k}\pd{}{q_{1,k}}+
    D_t^{k}(q_3)\pd{}{q_{2,k}}\right),\\
  \tilde{D}_t&=D_t + \sum_{l=0}^2q_{l,1}\pd{}{q_l} +
  \sum_{l=0}^2\sum_{k\ge1}q_{l,k+1}\pd{}{q_{l,k}},
\end{align*}
where~$q_3=(\alpha q_{0,1} - q_{2,1} + (3u_1\alpha - u_3)q + (3u\alpha -
2u_2)q_1 - 2u_1q_2)/u$.

To any cosymmetry~$\psi_i$ we associate a nonlocal form~$Q^i$ defined by the
equations
\begin{align*}
  \tilde D_x(Q^i)&=(\alpha\psi_i - D_x^2(\psi_i))q,\\
  \tilde D_t(Q^i)&=((u_2 - 3u\alpha)\psi_i + uD_x^2(\psi_i))q^u + (u_1\psi_i -
  uD_x(\psi_i))q_1 + u\psi_iq_2
  \\ &  - D_x(\psi_i)q_{0,1} + \psi_iq_{1,1}.
\end{align*}

\subsection{Nonlocal vectors}
\label{sec:nonlocal-vectors}

Nonlocal vectors arise in the $\ell^*$-covering. The latter is the extension
of the initial equation by the equation
\begin{equation*}
  -\alpha p_t +  p_{xxt} + up_{xxx} + u_xp_{xx} + (u_{xx} - 3\alpha u) p_x =0
\end{equation*}
with the total derivatives
\begin{align*}
  \tilde{D}_x&=D_x+\sum_{l=0}^2p_{l+1}\pd{}{p_l}+
  \sum_{k\ge1}\left(p_{1,k}\pd{}{p_{0,k}} + p_{2,k}\pd{}{p_{1,k}}+
    D_t^{k}(p_3)\pd{}{p_{2,k}}\right),\\
  \tilde{D}_t&=D_t+\sum_{l=0}^2p_{l,1}\pd{}{p_l} +
  \sum_{l=0}^2\sum_{k\ge1}p_{l,k+1}\pd{}{p_{l,k}},
\end{align*}
where $p_3=((3u\alpha - u_2)p_1 - u_1p_2 + \alpha p_{0,1} - p_{2,1})/u$.

To any symmetry~$\phi_i$ there corresponds a nonlocal vector~$P^i$ defined by
the relations
\begin{align*}
  (P^i)_x&= (\phi_i\alpha - D_x^2(\phi_i))p,\\
  (P^i)_t &= ((u_2 - 3u\alpha)\phi_i + u_1D_x(\phi_i) + uD_x^2(\phi_i))p
  - uD_x(\phi_i)p_1 + u\phi_ip_2 \\
  &- D_x(\phi_i)p_{0,1} + \phi_ip_{1,1}.
\end{align*}

\subsection{Recursion operators for symmetries}
\label{sec:recurs-oper-symm}

The defining equation for these operators is
\begin{equation*}
  \alpha \tilde{D}_t(\tens{R}) -\tilde{D}^2_x\tilde{D}_t(\tens{R})
  - u\tilde{D}_x^3(\tens{R}) - 2u_1\tilde{D}_x^2(\tens{R})
  + (3\alpha u - 2u_2)\tilde{D}_x(\tens{R}) + (3\alpha u_1 - u_3)\tens{R}=0
\end{equation*}
where the total derivatives are those presented in
Subsection~\ref{sec:nonlocal-forms}. We consider here two nontrivial
solutions:
\begin{align*}
  \tens{R}_{-1} &= (4Q^{3/2}(u\alpha - u_2 )^{-5/2}u(2u^2u_1\alpha^2 +
  uu_{0,1}\alpha^2 - 4uu_1u_2\alpha - uu_{2,1}\alpha \\
  & - u_{0,1}u_2\alpha + 2 u_1u_2^2 + u_2u_{2,1}) - 4q_2u^2(u\alpha -
  u_2)^{-1}
  + 4q_1(u\alpha - u_2 )^{-3}u( - 2u^2u_1\alpha^2 \\
  & - u u_{0,1}\alpha^2 + 4uu_1u_2\alpha + uu_{2,1}\alpha +
  u_{0,1}u_2\alpha - 2u_1u_2^2 - u_2u_{2,1}) \\
  & + 2q((u\alpha - u_2)^{-3}(4u^2 u_1^2\alpha^2 +
  4uu_{0,1}u_1\alpha^2 - 8uu_1^2u_2\alpha - 4uu_1u_{2,1}\alpha \\
  & + u_{0,1}^2\alpha^2 - 4u_{0,1}u_1u_2\alpha -
  2u_{0,1}u_{2,1}\alpha + 4u_1^2u_2^2 + 4u_1u_2u_{2,1}+ u_{2,1}^2) \\
  & + 2u^2\alpha(u\alpha- u_2)^{-1}))/u^2
  \intertext{and}
  \tens{R}_{3} &= Q^3u_1 - q_{1,1} - q_2u - q_1u_1 + q(2u\alpha - u_2).
\end{align*}
The operator corresponding to the second one is
\begin{equation*}
  \tens{R}_3= \alpha u_1D_x^{-1} - D_{xt} - uD_x^2  - u_1D_x + 2u\alpha-u_2.
\end{equation*}
The first operator is too complicated to present it here.

\subsection{Symplectic structures}
\label{sec:sympl-struct}

The equation defining symplectic structures is
\begin{equation*}
  -\alpha\tilde{D}_t(\tens{S}) + \tilde{D}_x^2D_t(\tens{S}) + u\tilde{D}_x^3(\tens{S})
  + u_1\tilde{D}_x^2(\tens{S}) + (u_2 - 3\alpha u)\tilde{D}_x(\tens{S}) =0,
\end{equation*}
where the total derivatives are from Subsection~\ref{sec:nonlocal-forms}. The
simplest solutions are
\begin{align*}
  \tens{S}_{5} &= Q^3,\\
  \tens{S}_{6} &= Q^4 + Q^3u - q_{0,1} - q_1u
\end{align*}
while the corresponding operators have the form
\begin{align*}
  \tens{S}_5& = \alpha D_x^{-1},\\
  \tens{S}_6& = D_x^{-1}\circ(\alpha u - u_2) + \alpha uD_x^{-1} - D_t - uD_x.
\end{align*}

\subsection{Hamiltonian structures}
\label{sec:hamilt-struct}

Hamiltonian structures are defined by the equation
\begin{equation*}
  \alpha\tilde{D}_t(\tens{H}) -\tilde{D}^2_x\tilde{D}_t(\tens{H}) -
  u\tilde{D}_x^3(\tens{H}) - 2u_1\tilde{D}_x^2(\tens{H})
  + (3\alpha u - 2u_2)\tilde{D}_x(\tens{H}) + (3\alpha u_1 - u_3)\tens{H}=0,
\end{equation*}
where the total derivatives are given in
Subsection~\ref{sec:nonlocal-vectors}. We present here three solutions. Two of
them are local and the third one is nonlocal:
\begin{align*}
  \tens{H}_{-3} &= p_1,\\
  \tens{H}_{-2} &=  - p_{0,1} - p_1u + pu_1,\\
  \tens{H}_1 &= P^1u_1 - p_{1,2} - p_{2,1}u + p_{0,1}u\alpha -
  p_2u_{0,1} + p_1u( - u\alpha + u_2) \\
  & + p( - uu_1\alpha + u_{0,1}\alpha + u_1u_2)
\end{align*}
The corresponding operators are
\begin{align*}
  \tens{H}_{-3}& = D_x,\\
  \tens{H}_{-2}& =  - D_t - uD_x + u_1,\\
  \tens{H}_1&= u_1D_x^{-1}\circ(( - 2uu_1\alpha - u_{0,1}\alpha + 2u_1u_2 +
  u_{2,1})/u) - uD_x^2D_t - D_xD_t^2 +  u\alpha D_t \\
  & - u_{0,1}D_x^2 + u( - u\alpha + u_2)D_x - 4uu_1\alpha + uu_3 + 3u_1u_2 +
  u_{2,1}.
\end{align*}

\subsection{Recursion operators for cosymmetries}
\label{sec:recurs-oper-cosymm}

These recursion operators are defined by the equation
\begin{align*}
  -\alpha\tilde{D}_t(\hat{\tens{R}}) + \tilde{D}_x^2D_t(\hat{\tens{R}})
  + u\tilde{D}_x^3(\hat{\tens{R}}) + u_1\tilde{D}_x^2(\hat{\tens{R}})
  + (u_2 - 3\alpha u)\tilde{D}_x(\hat{\tens{R}}) =0,
\end{align*}
where the total derivatives are given in
Subsection~\ref{sec:nonlocal-vectors}.  One of the solutions is
\begin{equation*}
  \hat{\tens{R}}_3 = P^1 + p_{1,1} + p_2u  - 2u\alpha + u_2
\end{equation*}
to which the operator
\begin{equation*}
  \hat{\tens{R}}_3= D_x^{-1}\circ(( - 2uu_1\alpha - u_{0,1}\alpha +
  2u_1u_2 + u_{2,1})/u) + D_{xt} + uD_x^2 - 2u\alpha + u_2
\end{equation*}
corresponds.

\section{Conclusion}
\label{sec:conclusion}

We finish this paper with a number of remarks.

\begin{remark}
  \label{rem:what-has-been}
  First of all, let us stress again what has been done above. We treated the
  Camassa--Holm equation directly, without artificial assumptions about its
  evolution or pseudo-evolu\-tion nature. The pass from the original
  form~\eqref{eq:1} to system~\eqref{eq:25} (which is \emph{not} in evolution
  form) was done by technical reasons only. On this way, we found an infinite
  family of pair-wise compatible Hamiltonian structures, recursion operators
  for symmetries and cosymmetries, and symplectic operators. These structures
  lead to existence of two commutative series of local symmetries and
  conservation laws and thus the equation is integrable.
\end{remark}

\begin{remark}
  \label{rem:nonlocal}
  Several comments are worth to be made in relation to
  Theorem~\ref{thm:brackets}. Two cases must be distinguished: local and
  nonlocal. In both cases the problem reduces to \emph{reconstruction} of
  nonlocal shadows up to symmetries and subsequent computation of the bracket.

  \paragraph*{Local case.}
  \label{sec:local-case}

  This case is simple and the result is as follows: any shadow in the
  $\ell$-covering corresponding to a recursion operator can be canonically
  lifted up to a symmetry of the covering equation and the commutator of these
  symmetries corresponds to the Nijenhuis bracket of the operators. A similar
  result is valid for the shadows in the $\ell^*$-covering corresponding to
  local bivectors (in this case the commutator of symmetries is related to the
  Schouten bracket).

  \paragraph*{Nonlocal case.}
  \label{sec:nonocal-case}

  The case of nonlocal operators is much more complicated and rests on the
  problem of how to commute nonlocal shadows. Seemingly there exists no
  \emph{natural} definition of such a commutator, but in
  Ref.~\cite{G-K-V-MIIGA} we proposed a procedure both to reconstruct and
  commute shadows. A weak point of this procedure is that it is based on
  rather cumbersome and intuitively not obvious notion of shadow
  equivalence. Dealing with equivalence classes of shadows necessitates
  elimination of certain kind of nonlocal variables that we call
  \emph{pseudo-constants} and that are intrinsically related to the covering
  at hand. This not always simple to do and we plan to simplify and clarify
  the procedure.
\end{remark}

\begin{remark}
  \label{rem:comp-two-repr}
  Of course, all invariants of the Camassa--Holm equation in its initial form
  can be obtained from those computed for the system by simple
  transformations. E.g., to obtain symmetries of Eq.~\eqref{eq:27} from those
  of~\eqref{eq:25} one should substitute in the $u$-component the variables
  $w_l$ by $D_x^l(u\alpha - u_2)$.  On the other way, any cosymmetry in the
  scalar case can be constructed directly from the corresponding cosymmetry of
  Eq.~\eqref{eq:25} by changing the variables~$w_l$ in the $w$-component
  by~$D_x^l(u\alpha - u_2)$. Similar transformations are applicable to other
  structures.
\end{remark}

\begin{remark}
  \label{rem:V-trick}
  Computation of symmetries and cosymmetries in the matrix representation can
  be simplified using the following observation:
  \begin{proposition}
    The following correspondence between the components of symmetries and
    cosymmetries of matrix equation~\eqref{eq:25} is valid\textup{:}
    \begin{equation*}\xymatrixcolsep{3.5pc}
      \xymatrix{
        \psi_k^w\ar@{|-{>}}[r]^{D_t+uD_x-u_1}&\psi_k^u\ar@{|-{>}}[r]^{\cdot(-1)}&
        \phi_{k-2}^u\ar@{|-{>}}[r]^{\alpha-D_x^2}&\phi_{k-2}^w.
      }
    \end{equation*}
  \end{proposition}
\end{remark}

\begin{remark}
  \label{rem:local-ov-hier}
  Now we shall show how to prove locality of symmetry hierarchies. Let us
  introduce a new variable~$v=w^2$. Then Eq.~\eqref{eq:25} transforms to
  \begin{equation}
    \label{eq:29}
    \begin{cases}
      v_t&=-u_xv-uw_x,\\
      v^2&=\alpha u-u_{xx},
    \end{cases}
  \end{equation}
  i.e., the first equation acquires potential form\footnote{This fact is
    related to existence of the nonlocal variable~$s_{1/2}$, see above.}. This
  means that for any symmetry~$\phi=(\phi^v,\phi^u)$ the
  form~$\omega_\phi=\phi^v\,dx-(u\phi^v+v\phi^u)\,dt$ is a conservation law of
  Eq.~\eqref{eq:29}. Comparing gradings, it is easily checked that all
  conservation laws of the form~$\omega_\phi$ are trivial and
  consequently~$\phi^v$ lies in the image of~$D_x$ for any
  symmetry~$\phi$. From this fact it follows that the action of recursion
  operators is local.
\end{remark}

\begin{remark}
  \label{sec:relations-n-m}
  Finally, we indicate relations between the Camassa--Holm equation and the
  equation describing short capillary-gravity waves (the Neveu--Manna
  equation, see~\cite{B,N-M})
  \begin{equation*}
    u_{xy }=u-u u_{xx}-\frac{1}{2}u_x^2+\frac{\lambda}{2} u_{xx} u_x^2.
  \end{equation*}
  We strongly believe that there exists a deformation connecting these two
  equations and their integrable structures are closely related to each
  other. We intend to discuss this relation elsewhere.
\end{remark}

\end{document}